\documentclass[11pt]{article}

\usepackage[left=1in, right=1in, top=1in, bottom=1in]{geometry}
\usepackage[group-separator={,}]{siunitx}
\usepackage{amsmath,amsfonts,amssymb, amsthm}
\usepackage{graphicx}
\usepackage{url}
\usepackage{mathtools}
\usepackage[textsize=tiny,textwidth=1cm,shadow]{todonotes}
\usepackage{thmtools}
\usepackage{enumitem}
\usepackage{pdflscape}
\usepackage{verbatim}
\usepackage{nccmath}
\usepackage[utf8]{inputenc}
\usepackage{array}
\definecolor{MyBlue}{rgb}{0.12, 0.12, 0.76}
\usepackage{algorithm}
\usepackage{algpseudocode}
\usepackage{algorithmicx}
\let\oldReturn\Return
\renewcommand{\Return}{\State\oldReturn}
\algtext*{EndWhile}% Remove "end while" text
\algtext*{EndIf}% Remove "end if" text
\algtext*{EndForAll}% Remove "end for" text
\algtext*{EndFor}% Remove "end for" text
\algtext*{EndFunction}% Remove "end function" text
\usepackage{pifont}% http://ctan.org/pkg/pifont
\usepackage{pgf}
\usepackage{tikz}
\usetikzlibrary{shadows,arrows,decorations,decorations.shapes,backgrounds,shapes,snakes,automata,fit,petri,shapes.multipart,calc,positioning,shapes.geometric,graphs,graphs.standard}
\makeatletter
\newcommand{\thickhline}{%
    \noalign {\ifnum 0=`}\fi \hrule height 1.4pt
    \futurelet \reserved@a \@xhline
}
\newcolumntype{"}{@{\hskip\tabcolsep\vrule width 1.4pt\hskip\tabcolsep}}
\makeatother
\allowdisplaybreaks[1] % let align break over multiple pages

\newtheorem{theorem}{Theorem}[section]
\newtheorem{lemma}{Lemma}[section]
\newtheorem{corollary}{Corollary}[theorem]
\newtheorem{definition}{Definition}[section]

\usepackage{natbib}
\usepackage{tablefootnote} 
\usepackage{subfigure}
\usepackage{epsfig}

\newtheorem{example}[theorem]{Example}
\newtheorem{proposition}	[theorem]	{Proposition} 

\newlist{exlist}{enumerate}{1}
\setlist[exlist]{label=(\alph*)}

\newcommand{\perf}{\textsc{Perf}}

\newcommand{\eps}{\epsilon}

\newcommand{\sm}{\setminus}

\begin{document}

\title{Resource Augmentation\thanks{Chapter~4 of the book {\em Beyond the
      Worst-Case Analysis of Algorithms}~\citep{bwca}.}}
\author{Tim Roughgarden\thanks{Department of Computer Science,
    Columbia University.  Supported in part by NSF award
    CCF-1813188 and ARO award W911NF1910294.  Email: \texttt{tim.roughgarden@gmail.com.}}}

\maketitle

\begin{abstract}
This chapter introduces {\em resource augmentation}, in which
the performance of an algorithm 
is compared
to the best-possible solution that is handicapped by less resources.
We consider three case studies: online paging, with cache size as the
resource; 
selfish routing, with capacity as the resource; and scheduling, with
processor speed as the resource.  Resource augmentation bounds also
imply ``loosely competitive'' bounds, which show that an algorithm's
performance is near-optimal for most resource levels.
\end{abstract}

\section{Online Paging Revisited}\label{s:paging}

This section illustrates the idea of resource augmentation with a
familiar example, the competitive analysis of online paging
algorithms.  Section~\ref{s:disc} discusses the pros and cons of
resource augmentation more generally, Sections~\ref{s:sr}
and~\ref{s:sched} describe additional case studies in routing and
scheduling, and Section~\ref{s:lc} shows how resource augmentation bounds
lead to ``loosely competitive'' guarantees.

\subsection{The Model}

Our first case study of resource augmentation concerns the online
paging problem introduced in Chapter~1.
Recall the ingredients of the problem:
\begin{itemize}

\item There is a slow memory with $N$ pages.

\item There is a fast memory (a {\em cache}) that can hold only
$k < N$ of the pages at a time.

\item Page requests arrive online over time, with one request per time
  step.  The decisions of an online algorithm at time~$t$ can depend
  only on the requests arriving at or before time~$t$.

\item If the page $p_t$ requested at time $t$ is already in the cache,
no action is necessary.

\item If $p_t$ is not in the cache, it must be brought in;
if the cache is full, one of its~$k$ pages must be evicted.
This is called a {\em page fault}.%, and one unit of cost is incurred in
%this case.
\footnote{This model corresponds to ``demand paging,''
meaning algorithms that modify the cache only in response to a page
fault.  The results in this section continue to hold in the more
general model in which
an algorithm is allowed to make arbitrary changes to the
cache at each time step, whether or not there is a page fault,
with the cost incurred by the algorithm equal to the number of changes.}

\end{itemize}
We measure the performance $\perf(A,z)$ of an algorithm~$A$ on a page
request sequence~$z$ by the number of page faults incurred.

\subsection{FIF and LRU}

As a benchmark, what would we do if we had clairvoyance about all
future page requests?  An intuitive greedy algorithm minimizes the
number of page faults.
\begin{theorem}[\citet{B67}]\label{t:fif}
  The {\em Furthest-in-the-Future (FIF)} algorithm, which on a page fault
  evicts the page to be requested furthest in the future, always
  minimizes the number of page faults.
\end{theorem}

The FIF algorithm is not an online algorithm, as its eviction
decisions depend on future page requests.  The {\em Least Recently
  Used (LRU)} policy, which on a page fault evicts the page whose most
recent request is furthest in the past, is an online surrogate for the
FIF algorithm that uses the past as an approximation for the future.
Empirically, the LRU algorithm performs well on most ``real-world''
page request sequences---not much worse than the unimplementable FIF
algorithm, and better than other online algorithms such as first-in
first-out (FIFO).  The usual explanation for the superiority of the
LRU algorithm is that the page request sequences that arise in
practice exhibit locality of reference, with recent requests likely to
be requested again soon, and that LRU automatically adapts to and
exploits this locality.

\subsection{Competitive Ratio}\label{ss:cr}

One popular way to assess the performance of an online algorithm is
through its competitive ratio:\footnote{See Chapter~24 for a
  deep dive on alternatives to worst-case analysis in the
  competitive analysis of online algorithms.}
\begin{definition}[\citet{ST85}]\label{d:cr}
The {\em competitive
ratio} of an online algorithm~$A$ is its worst-case performance
(over inputs~$z$) relative to an optimal {\em offline} algorithm $OPT$
that has advance knowledge of the entire input:
\[
\max_z \frac{\perf(A,z)}{\perf(OPT,z)}.
\]
\end{definition}
For the objective of minimizing the number of page faults,
the competitive ratio is always at least~1, and the closer to~1 the
better.\footnote{One usually ignores any extra additive terms in the
competitive ratio, which vanish as $\perf(OPT,z) \rightarrow \infty$.} 

Exercise~1.1 of Chapter~1 shows that, for every deterministic online
paging algorithm~$A$ and cache size~$k$, there are arbitrarily long
page request sequences~$z$ such that~$A$ faults at every time step
while the FIF algorithm faults at most once per~$k$ time steps.  This
example shows that every deterministic online paging algorithm has
a competitive ratio of at least~$k$.  For most natural online algorithms,
there is a matching upper bound of~$k$.  This state of affairs is
unsatisfying for several reasons:
\begin{enumerate}

\item The analysis gives an absurdly pessimistic performance
  prediction for LRU (and all other deterministic online algorithms),
  suggesting that a 100\% page fault rate is unavoidable.

\item The analysis suggests that online algorithms perform worse
  (relative to FIF) as the cache size grows, a sharp departure from
  empirical observations.

\item The analysis fails to differentiate between competing policies
  like LRU and FIFO, which both have a competitive ratio of~$k$.

\end{enumerate}
We next address the first two issues through a resource augmentation
analysis (but not the third, see Exercise~\ref{exer:fifo}).

\subsection{A Resource Augmentation Bound}

In a {\em resource augmentation} analysis, the idea is to compare the
performance of a protagonist algorithm (like LRU) to an all-knowing
optimal algorithm that is handicapped by ``less resources.''
Naturally, weakening the capabilities of the offline optimal algorithm
can only lead to better approximation guarantees.

Let~$\perf(A,k,z)$ denote the number of page faults incurred by the
algorithm~$A$ with cache size~$k$ on the page request sequence~$z$.
The main result of this section is:
\begin{theorem}[\citet{ST85}]\label{t:lru_ra}
For every page request sequence~$z$ and cache sizes $h \le k$,
\[
\perf(LRU,k,z) \le \frac{k}{k-h+1} \cdot \perf(FIF,h,z),
\]
plus an additive error term that goes to~0 with $\perf(FIF,h,z)$.
\end{theorem}
For example, LRU suffers at most twice as many page faults
as the unimplementable FIF algorithm
when the latter has roughly half the cache size.

\begin{proof}
Consider an arbitrary page request sequence $z$ and cache sizes $h \le
k$.  We first prove an upper bound on the number of page faults incurred by
the LRU algorithm, and then a lower bound on the number of faults
incurred by the FIF algorithm.
A useful idea for accomplishing both goals
is to break $z$ into {\em blocks}
$\sigma_1,\sigma_2,\ldots,\sigma_b$.
Here $\sigma_1$ is the maximal prefix of $z$ in which
only $k$ distinct pages are requested; the block~$\sigma_2$ starts
immediately after and is maximal subject to only $k$ distinct pages
being requested within it;
%(ignoring what was requested in $\sigma_1$);
and so on.

For the first step, note that LRU faults at most $k$ times within
a single block---at most once per page requested in the block.  The
reason is that once a page is brought into the cache, LRU won't
evict it until $k$ other distinct pages are requested, and this can't
happen until the following block.  Thus LRU incurs at most $bk$ page
faults, where $b$ is the number of blocks.
See Figure~\ref{f:lru}(a).

\begin{figure}
\begin{center}
\mbox{\subfigure[Blocks of a request sequence]{\epsfig{file=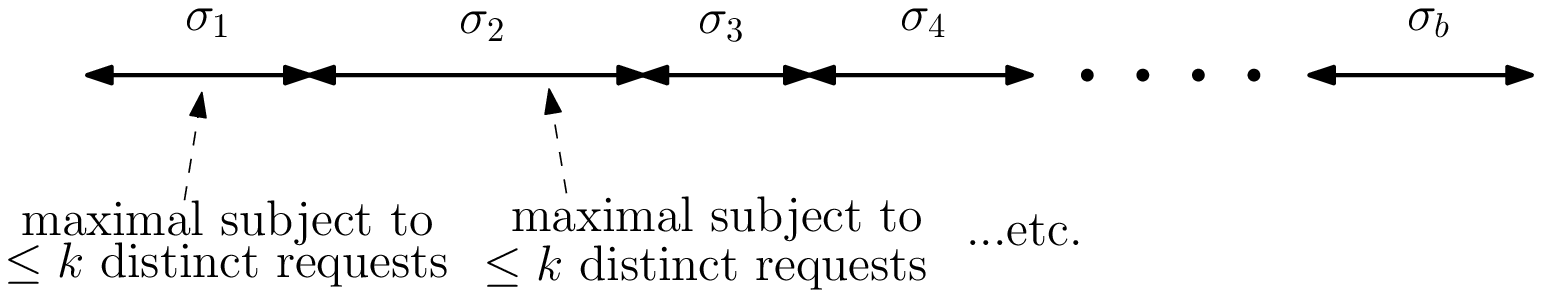,width=.475\textwidth}}\quad\quad
\subfigure[Lower bound for FIF (with $h=k$)]{\epsfig{file=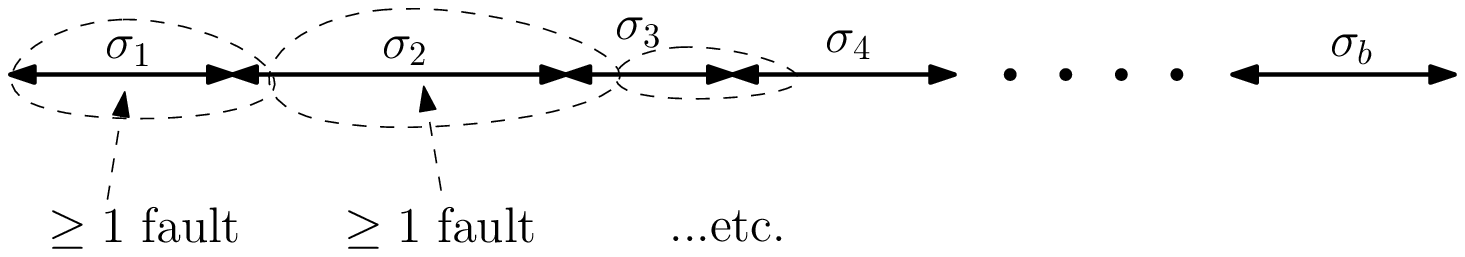,width=.475\textwidth}}}
\caption{Proof of Theorem~\ref{t:lru_ra}.  In~(a), the blocks of a page request sequence; the LRU algorithm incurs at most $k$ page faults in each.
In~(b), the FIF algorithm incurs at least $k-h+1$ page faults in each ``shifted block.''}\label{f:lru}
\end{center}
\end{figure}

For the second step, consider the FIF algorithm with a cache size $h
\le k$.  Consider the first block $\sigma_1$ plus the
first request of the second block $\sigma_2$.  Since $\sigma_1$ is
maximal, this 
represents requests for $k+1$ distinct pages.  At least $k-h+1$ of
these pages are initially absent from the size-$h$ cache, so no
algorithm can serve all $k+1$ pages without incurring at least $k-h+1$
page faults.
Similarly, suppose the first
request of~$\sigma_2$ is the page~$p$.  After
an algorithm serves the request for~$p$, the cache 
contains only $h-1$ pages other than~$p$.  By the maximality of~$\sigma_2$,
the ``shifted block'' comprising the rest of~$\sigma_2$ and the first
request of~$\sigma_3$ includes requests for~$k$ distinct pages other
than~$p$; these cannot all be
served without incurring another 
\[
\underbrace{k}_{\text{requests other than $p$}} - \underbrace{(h-1)}_{\text{pages in     cache other than $p$}}
\]
page faults.  And so on,
resulting in at least $(b-1)(k-h+1)$ page faults overall.
See Figure~\ref{f:lru}(b).

We conclude that 
\[
\perf(LRU,k,z) \le bk \le \frac{k}{k-h+1} \cdot \perf(FIF,h,z) + \frac{k}{(b-1)(k-h+1)}.
\]
The additive error term goes to~0 with~$b$, and the proof is complete.
\end{proof}

\section{Discussion}\label{s:disc}

Resource augmentation guarantees make sense for any problem in which
there is a natural notion of a ``resource,'' with algorithm performance
improving in the resource level; see Sections~\ref{s:sr}
and~\ref{s:sched} for two further examples.
In general, a resource augmentation guarantee implies that
the performance curves (i.e., performance as a function of
resource level) of an online algorithm and the offline optimal
algorithm are similar (Figure~\ref{f:ra}).

\begin{figure}
\begin{center}
\mbox{\subfigure[A good competitive ratio]{\epsfig{file=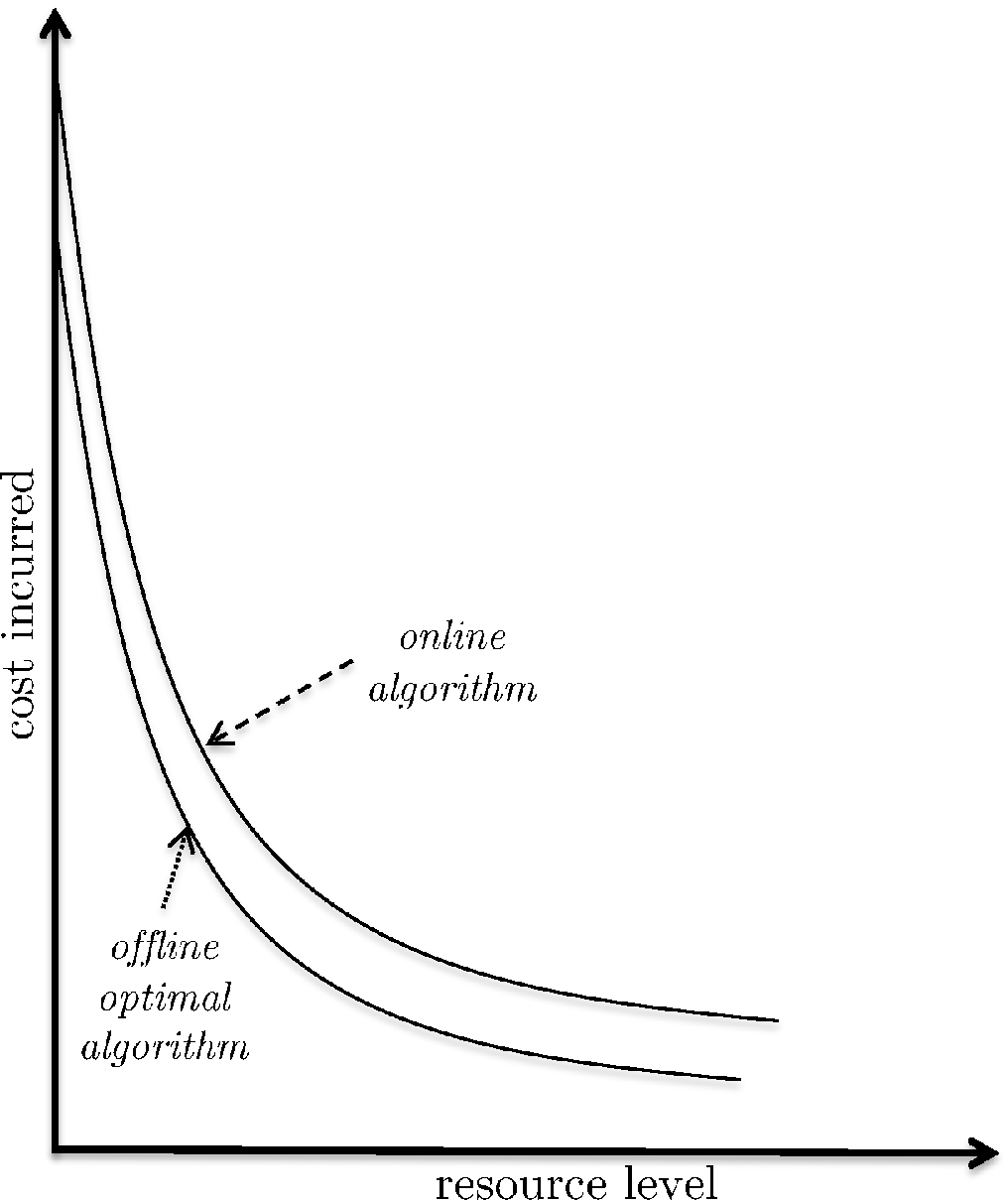,width=.45\textwidth}}\quad\quad
\subfigure[A resource augmentation guarantee]{\epsfig{file=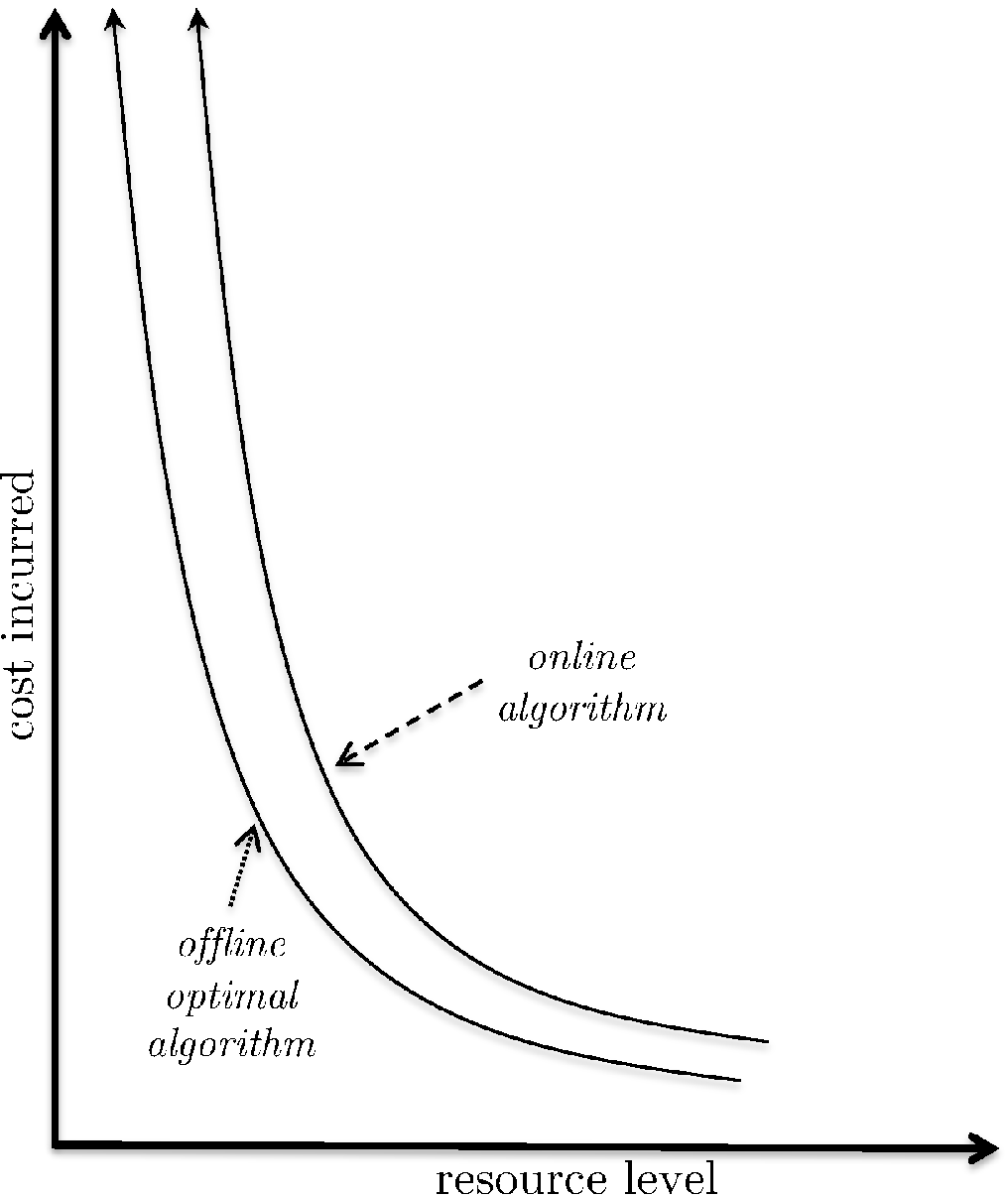,width=.45\textwidth}}}
\caption{Competitive ratio guarantees vs.\ resource augmentation
  guarantees.  All curves plot, for a fixed input, the cost incurred
  by an algorithm (e.g., number of page faults) as a function of the
  resource level (e.g., the cache size).  In~(a), a good upper bound
  on the competitive ratio requires that the curve for the online
  algorithm closely approximates that of the offline optimal algorithm
  pointwise over the $x$-axis.  In~(b), the vertical distance between
  the two curves (and the competitive ratio) grows
  large as the resource level approaches its minimum.
  A resource augmentation guarantee roughly translates to the relaxed
  requirement that every point of the online algorithm's performance
  curve has a nearby neighbor somewhere on the optimal offline
  algorithm's performance curve.}\label{f:ra}
\end{center}
\end{figure}

The resource augmentation guarantees in this chapter resemble
worst-case analysis, in that no model of data is proposed; the
difference is purely in the method of measuring algorithm performance
(relative to optimal performance).  As usual, this is both a feature
and a bug: the lack of a data model guarantees universal
applicability, but also robs the analyst of any opportunity to
articulate properties of ``real-world'' inputs that might lead to a
more accurate and fine-grained analysis.  There is nothing inherently
worst-case about resource augmentation guarantees, however, and the
concept can equally well be applied with one of the models of data
discussed in the other parts of this book.\footnote{For example,
  Chapter~27 combines robust distributional analysis with resource
  augmentation, in the context of prior-independent auctions.}

How should you interpret a resource augmentation guarantee like
Theorem~\ref{t:lru_ra}?  Should you be impressed?  Taken at face
value, Theorem~\ref{t:lru_ra} seems much more meaningful than the
competitive ratio of~$k$ without resource augmentation, even though it
doesn't provide particularly sharp performance predictions (as
to be expected, given the lack of a model of data).  But isn't it an
``apples vs.\ oranges'' comparison?  The optimal offline algorithm is
powerful in its knowledge of all future page requests, but it's
artificially hobbled by a small cache.

One interpretation of a resource augmentation guarantee is as a
two-step recipe for building a system in which an online algorithm
has good performance.
\begin{enumerate}

\item Estimate the resource level (e.g., cache size) such that the
  optimal offline algorithm has acceptable performance (e.g., page
  fault rate below a given target).\footnote{Remember: 
  competing with the optimal algorithm is only useful when its
  performance is good in some absolute sense!}
This task can be simpler than reasoning
simultaneously about the cache size and paging algorithm design
decisions.  

\item Scale up the resources to
realize the resource augmentation guarantee (e.g., doubling the cache size
needed by the FIF algorithm to achieve good performance).

\end{enumerate}

A second justification for resource
augmentation guarantees is that they usually lead directly to good ``apples
vs.\ apples'' comparisons for most resource levels (as suggested by
Figure~\ref{f:ra}(b)).
Section~\ref{s:lc} presents a detailed case study in the context of
online paging.

\section{Selfish Routing}\label{s:sr}

Our second case study of a resource augmentation guarantee concerns a
model of {\em selfish routing} in a congested network.

\subsection{The Model and a Motivating Example}

In selfish routing, we consider a directed flow network $G=(V,E)$,
with $r$ units of flow traveling from a source vertex $s$ to a sink
vertex $t$; $r$ is called the {\em traffic rate}.
Each edge~$e$ of the network has a flow-dependent cost
function $c_e(x)$.  For example, in the network in
Figure~\ref{fig:pigou}(a), the top edge has a constant cost function
$c(x) = 1$, while the cost to traffic on the bottom edge equals the
amount of flow~$x$ on the edge.

\begin{figure}
\begin{center}
\mbox{\subfigure[]{\epsfig{file=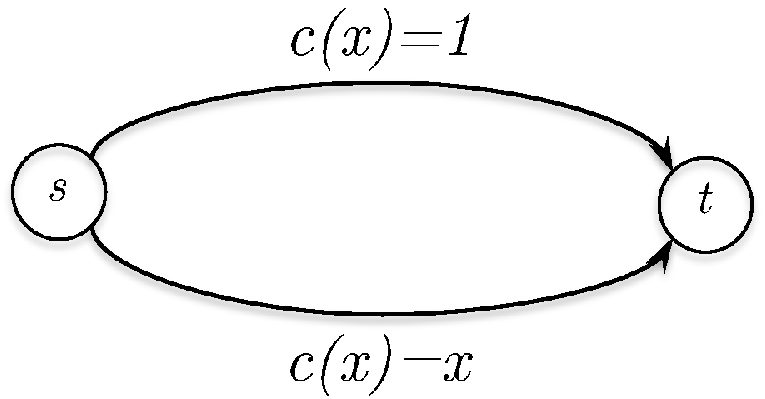,width=.45\textwidth}}\quad\quad
\subfigure[]{\epsfig{file=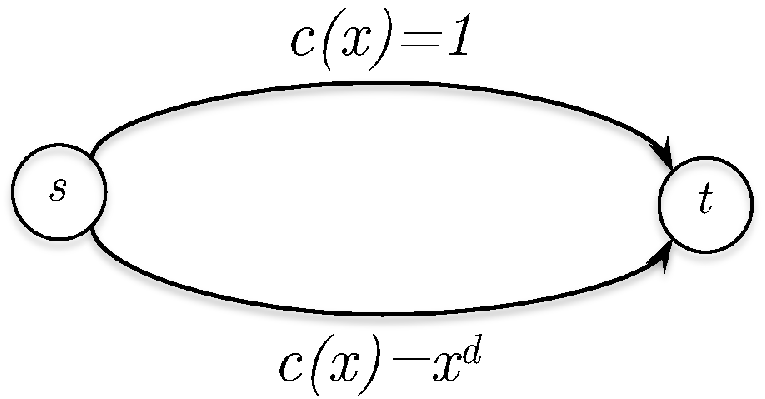,width=.45\textwidth}}}
\caption{Two selfish routing networks.
Each cost function $c(x)$ describes the cost incurred by users of an
edge, as a function of the amount of traffic routed on that
edge.}\label{fig:pigou}
\end{center}
\end{figure}

The key approximation concept in selfish routing networks is the {\em
  price of anarchy} which, as usual with approximation ratios, is
defined as the ratio between two things: a realizable protagonist and
a hypothetical benchmark.  

Our protagonist is an {\em equilibrium flow}, in which all traffic is
routed on shortest paths, where the length of an $s$-$t$ path~$P$ is
the (flow-dependent) quantity $\sum_{e \in P} c_e(f_e)$, where $f_e$
denotes the amount of flow using the edge~$e$.  In Figure~\ref{fig:pigou}(a),
with one unit of traffic, the only equilibrium flow sends all traffic
on the bottom edge.  If $\eps > 0$ units of traffic were routed on the
top path, that traffic would not be routed on a shortest path
(incurring cost~1 instead of~$1-\eps$), and hence would want to switch
paths.

Our benchmark is the optimal solution, meaning the fractional
$s$-$t$ flow that routes the $r$ units of traffic to minimize the
total cost $\sum_{e \in E} c_e(f_e)f_e$.
For example, in Figure~\ref{fig:pigou}(a), the optimal
flow splits traffic evenly between the two paths, for a cost of
$\tfrac{1}{2} \cdot 1 + \tfrac{1}{2} \cdot \tfrac{1}{2} =
\tfrac{3}{4}$.
The cost of the equilibrium flow is $0 \cdot 1 + 1 \cdot 1 = 1$.

The price of anarchy of a selfish routing network is defined as the
ratio between the cost of an equilibrium flow and that of an optimal
flow.\footnote{It turns out that the equilibrium flow cost is uniquely
  defined in every selfish routing network with continuous and
  nondecreasing edge cost functions; see the Notes for details.\label{foot:unique}}  In
the network in Figure~\ref{fig:pigou}(a), the price of anarchy is $4/3$.

An interesting research goal is to identify selfish routing networks
in which the price of anarchy is close to~1---networks in which
decentralized optimization by selfish users performs almost as well as
centralized optimization.  Unfortunately, without any restrictions on
edges' cost functions, the price of anarchy can be arbitrarily large.
To see this, replace the cost function on the bottom edge in
Figure~\ref{fig:pigou}(a) by the function $c(x) = x^d$ for a large
positive integer~$d$ (Figure~\ref{fig:pigou}(b)).  The equilibrium
flow and its cost remain the
same, with all selfish traffic using the bottom edge for an overall
cost of~1.  The optimal flow, however, improves with $d$: Routing
$1-\epsilon$ units of flow on the bottom edge and $\epsilon$ units on
the top edge yields a flow with cost $\epsilon+ (1-\epsilon)^{d+1}$.
This cost tends to~0 as $d$ tends to infinity and $\epsilon$ tends
appropriately to~0, and hence the price of anarchy goes to infinity
with~$d$.

\subsection{A Resource Augmentation Guarantee}

Despite the negative example above, a very general resource
augmentation guarantee holds in selfish routing networks.\footnote{This
  result holds still more generally, in networks with multiple source and
  sink vertices (Exercise~\ref{exer:mc}).}
\begin{theorem}[\citet{RT00}]\label{thm:rt}
  For every network~$G$ with nonnegative, continuous, and
  nondecreasing cost functions, for every traffic rate $r > 0$, and
  for every $\delta > 0$, the cost of an equilibrium flow in~$G$ with
  traffic rate~$r$ is at most $\tfrac{1}{\delta}$ times the cost of an
  optimal flow with traffic rate $(1+\delta)r$.
\end{theorem}

For example, consider the network in Figure~\ref{fig:pigou}(b) with
$r = \delta = 1$ (and large~$d$).  The cost of the equilibrium flow
with traffic rate~1 is~1.  The optimal flow can route one unit of
traffic cheaply (as we've seen), but then the network gets clogged
up and it has no choice but to incur one unit of cost on the second
unit of flow (the best it can do is route it on the top edge).  Thus
the cost of an optimal flow with double the traffic exceeds that of
the original equilibrium flow.

Theorem~\ref{thm:rt} can be reformulated 
as a comparison between an
equilibrium flow in a network with ``faster'' edges and an optimal flow
in the original network.  For example, simple calculations
(Exercise~\ref{exer:bicrit}) show that the following statement is
equivalent to Theorem~\ref{thm:rt} with $\delta=1$.
\begin{corollary}\label{cor:bicrit}
  For every network~$G$ with nonnegative, continuous, and
  nondecreasing cost functions and for every traffic rate $r > 0$, the
cost of an equilibrium flow in~$G$ with traffic rate~$r$ and cost
functions~$\{ \tilde{c}_e \}_{e \in E}$ is at most that of
an optimal flow in~$G$ with traffic rate $r$ and cost functions $\{
c_e \}_{e \in E}$, where each function $\tilde{c}_e$ is derived from
$c_e$ as $\tilde{c}_e(x) = c_e(x/2)/2$.
\end{corollary}

Corollary~\ref{cor:bicrit} takes on a particularly appealing form
in networks with M/M/1 delay functions, meaning cost functions
of the form $c_e(x) = 1/(u_e-x)$, where $u_e$ can be interpreted as
an edge capacity or a queue service rate.  (If $x \ge u_e$,
interpret~$c_e(x)$ as $+\infty$.)  In this case,
the modified function $\tilde{c}_e$ in Corollary~\ref{cor:bicrit} is
\[
\tilde{c}_e(x) = \frac{1}{2(u_e - \tfrac{x}{2})} = \frac{1}{2u_e - x}.
\]  
Corollary~\ref{cor:bicrit} thus translates to the following design
principle for selfish routing networks with M/M/1 delay functions:
to outperform optimal routing, double the capacity of every edge.

\subsection{Proof of Theorem~\ref{thm:rt} (Parallel Edges)}\label{ss:rtpf1}

As a warm-up to the proof of Theorem~\ref{thm:rt}, consider the special
case where~$G=(V,E)$ is a network of parallel edges, meaning $V =
\{s,t\}$ and every edge of~$E$ is directed from~$s$ to~$t$ (as in
Figure~\ref{fig:pigou}).  Choose a traffic rate~$r > 0$; a
cost function~$c_e$ for each edge $e \in
E$ that is nonnegative, continuous, and nondecreasing; and the
parameter~$\delta > 0$.  Let~$f$ and~$f^*$ denote equilibrium
and optimal flows in~$G$ at traffic rates~$r$ and~$(1+\delta)r$,
respectively.  
The equilibrium flow~$f$ routes traffic only on shortest paths, so there
is a number~$L$ (the shortest $s$-$t$ path length) such that 
\[
\begin{array}{cl}
c_e(f_e)  = L & \text{if $f_e > 0$;}\\
c_e(f_e) \ge L & \text{if $f_e = 0$.}
\end{array}
\]
The cost of the equilibrium flow~$f$ is then
\[
\sum_{e \in E} c_e(f_e)f_e
= \sum_{e \in E \,: \, f_e > 0} c_e(f_e)f_e
= \sum_{e \in E \,: \, f_e > 0} L \cdot f_e = r \cdot L,
\]
as the total amount of flow~$\sum_{e \,:\, f_e > 0}  f_e$ equals the
traffic rate~$r$.

How can we bound from below the cost of the optimal flow~$f^*$,
relative to the cost~$rL$ of~$f$?  To proceed, bucket the edges of~$E$
into two categories: 
\[
\begin{array}{cl}
E_1 &\text{:= the edges~$e$ with $f^*_e \ge f_e$;}\\
E_2 &\text{:= the edges~$e$ with $f^*_e < f_e$.}
\end{array}
\]
With so few assumptions on the network cost functions, we can't say
much about the costs of edges under the optimal flow~$f^*$.  The two
things we {\em can} say are that $c_e(f^*_e) \ge L$ for all
$e \in E_1$ (because cost functions are nondecreasing) and that
$c_e(f^*_e) \ge 0$ for all $e \in E_2$ (because cost functions are
nonnegative).  At the very least, we can therefore lower bound the
cost of~$f^*$ by
\begin{equation}\label{eq:sr1}
\sum_{e \in E} c_e(f^*_e)f^*_e \ge 
\sum_{e \in E_1} c_e(f^*_e)f^*_e \ge 
L \cdot \sum_{e \in E_1} f^*_e.
\end{equation}
How little traffic could~$f^*$ possibly route on the edges of~$E_1$?
The flow routes $(1+\delta)r$ units of traffic overall.  It routes
less flow than~$f$ on the edges of~$E_2$ (by the definition of~$E_2$),
and~$f$ routes at most~$r$ units (i.e., its full traffic rate) on these
edges.  Thus
\begin{equation}\label{eq:sr2}
\sum_{e \in E_1} f^*_e 
= (1+\delta)r - \sum_{e \in E_2} f^*_e
\ge (1+\delta)r - \underbrace{\sum_{e \in E_2} f_e}_{\le r}
\ge \delta r.
\end{equation}
Combining the inequalities~\eqref{eq:sr1} and~\eqref{eq:sr2} shows
that the cost of~$f^*$ is at least $\delta \cdot rL$, which is
$\delta$ times the cost of~$f$, as desired.

\subsection{Proof of Theorem~\ref{thm:rt} (General Networks)}

Consider now the general case of Theorem~\ref{thm:rt}, in which the
network~$G=(V,E)$ is arbitrary.  General networks are more complex
than networks of parallel edges because there is no longer a
one-to-one correspondence between edges and paths---a path might
comprise many edges, and an edge might participate in many different
paths.  This complication aside, the proof proceeds similarly to that
for the special case of networks of parallel edges.

Fix a traffic rate~$r$, a cost function~$c_e$ for each edge $e \in E$,
and the parameter $\delta > 0$.  As before, let~$f$ and~$f^*$ denote
equilibrium and optimal flows in~$G$ at traffic rates~$r$
and~$(1+\delta)r$, respectively.  It is still true that there is a
number~$L$ such that all traffic in~$f$ is routed on paths~$P$ with
length $\sum_{e \in P} c_e(f_e)$ equal to~$L$, and such that all
$s$-$t$ paths have length at least~$L$.  The cost of the equilibrium
flow is again~$rL$.

The key trick in the proof is to replace, for the sake of analysis,
each cost function $c_e(x)$ (Figure~\ref{fig:bicrit}(a)) by the larger
cost function $\bar{c}_e(x) = \max\{ c_e(x), c_e(f_e) \}$
(Figure~\ref{fig:bicrit}(b)).  This trick substitutes for the
decomposition in Section~\ref{ss:rtpf1} of~$E$ into~$E_1$ and~$E_2$.
With the fictitious cost functions $\bar{c}_e$, edge costs are always
as large as if the equilibrium flow~$f$ had already been routed in the
network.

\begin{figure}[t]
\begin{center}
\mbox{\subfigure[Graph of cost function $c_e$ and
its value at flow value
$f_e$]{\epsfig{file=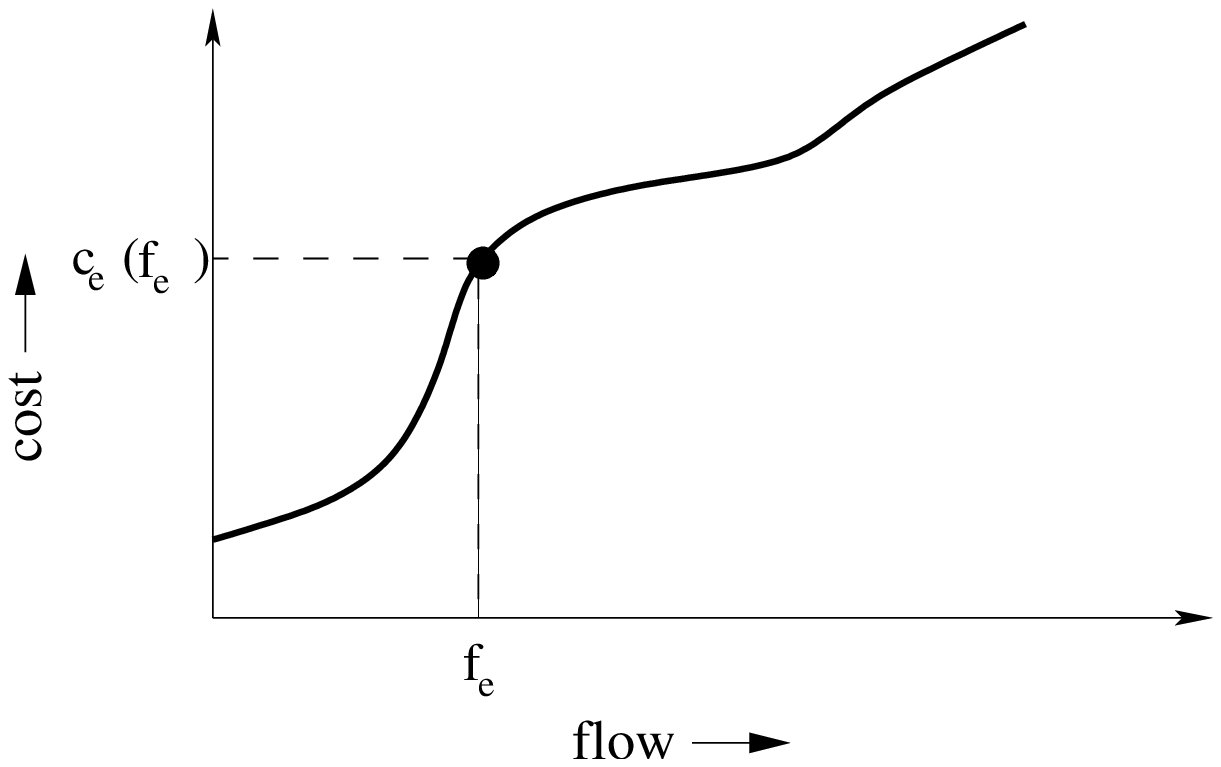,width=.40\textwidth}}\qquad\qquad
\subfigure[Graph of cost function
$\bar{c}_e$]{\epsfig{file=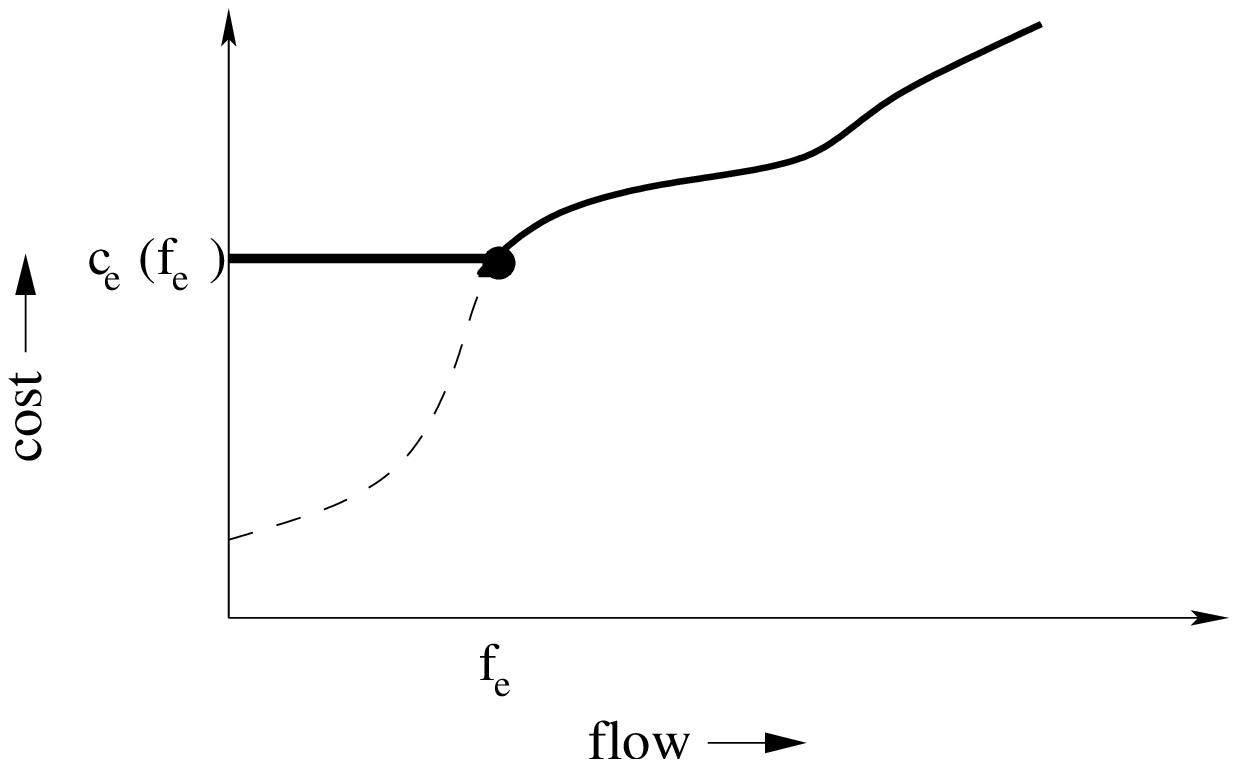,width=.40\textwidth}}}
\caption[Proof of Theorem~\ref{thm:bicrit}]{Construction in the proof
of Theorem~\ref{thm:rt} of the fictitious cost function $\bar{c}_e$ 
from the original cost function $c_e$ and equilibrium flow value $f_e$.}\label{fig:bicrit}
\end{center}
\end{figure}

By design, the cost of the optimal flow~$f^*$ is easy to bound from
below with the fictitious cost functions.  
Even with zero flow in the network,
every $s$-$t$ path has cost at least~$L$ with respect to these functions.  
Because~$f^*$ routes $(1+\delta)r$ units of traffic on
paths with (fictitious) cost at least~$L$, its total (fictitious) cost
with respect to the $\bar{c}_e$'s is at least $(1+\delta)rL$.

We can complete the proof by showing that the fictitious cost of $f^*$
(with respect to the $\bar{c}_e$'s) exceeds its real cost
(with respect to the $c_e$'s) by at most $rL$, the equilibrium flow
cost.  For each edge~$e \in E$ and $x \ge 0$, 
$\bar{c}_e(x) - c_e(x)$ is either~0 (if $x \geq f_e$)
or bounded above by $c_e(f_e)$ (if $x < f_e$); in any case,   
\[
\underbrace{\bar{c}_e(f^*_e)f^*_e}_{\text{fictitious cost of~$f^*$
    on~$e$}} - \underbrace{c_e(f^*_e)f^*_e}_{\text{real cost of~$f^*$
    on~$e$}} \le \underbrace{c_e(f_e)f_e}_{\text{real cost of~$f$ on~$e$}}.
\]
Summing this inequality over all edges $e \in E$ shows
that the difference between the costs of~$f^*$ with respect to the
different cost functions is at most the cost of~$f$ (i.e., $rL$);
this completes the proof of Theorem~\ref{thm:rt}.

\section{Speed Scaling in Scheduling}\label{s:sched}

The lion's share of killer applications of resource augmentation
concern scheduling problems.  This section describes one paradigmatic
example.

\subsection{Non-Clairvoyant Scheduling}

We consider a model with a single machine and
$m$ jobs that arrive online.
Each job~$j$ has a release time $r_j$ and the algorithm is unaware of
the job before this time.
Each job~$j$ has a processing time~$p_j$,
indicating how much machine time is necessary to complete it.
We assume that preemption is allowed, meaning that a job can be
stopped mid-execution and restarted from the same point (with no loss)
at a subsequent time.

We consider the basic objective of minimizing the total flow
time:\footnote{This objective is also called the total response time.}
\[
\sum_{j=1}^m \left( C_j - r_j \right),
\]
where $C_j$ denotes the completion time of job~$j$.
For an alternative formulation, note that each infinitesimal time interval
$[t,t+dt]$ contributes $dt$ to the flow time~$C_j-r_j$ of every job
that is {\em active} at time~$t$, meaning released but not yet
completed.  Thus, the total flow time can be written as
\begin{equation}\label{eq:obj2}
\int_0^{\infty} |X_t|dt,
\end{equation}
where~$X_t$ denotes the active jobs at time~$t$.

The {\em shortest
remaining processing time (SRPT)} algorithm always processes the
job that is closest to completion (preempting jobs as needed).
This algorithm makes~$|X_t|$ as small as possible for all times~$t$
(Exercise~\ref{exer:srpt}) and is therefore optimal.  This is a rare
example of a problem where the optimal offline algorithm is
implementable as an online algorithm.

SRPT uses knowledge of the job processing times to make decisions, and as such
is a {\em clairvoyant} algorithm.  What about applications in which
a job's processing time is not known before it completes, where a {\em
  non-clairvoyant} algorithm is called for?
No non-clairvoyant online algorithm can guarantee a total flow time
close to that achieved by SRPT (Exercise~\ref{exer:nonclairvoyant}).
Could a resource augmentation approach provide more helpful
algorithmic guidance?

\subsection{A Resource Augmentation Guarantee for SETF}\label{ss:setf}

The natural notion of a ``resource'' in this scheduling problem is
processor speed.  Thus, a resource augmentation guarantee would assert
that the total flow time of some non-clairvoyant protagonist {\em with
  a faster machine} is close to that of SRPT with the original
machine.

We prove such a guarantee for the {\em shortest elapsed time first
  (SETF)} algorithm, which always processes the job that has been
processed the least so far.  When multiple jobs are tied for the
minimum elapsed time, the machine splits its processing power equally
between them.  SETF does not use jobs' processing times to make
decisions, and as such is a non-clairvoyant algorithm.

\begin{example}\label{ex:setf}
Fix parameters~$\eps,\delta > 0$, with $\delta$ much smaller than
$\eps$.  With an eye toward a resource augmentation guarantee, we
compare the total flow time of SETF with a machine with
speed~$1+\eps$---meaning that the machine can process~$(1+\eps)t$ units of
jobs in a time interval of length~$t$---to that of SRPT with a unit-speed
machine.

Suppose~$m$ jobs arrive at times~$r_1=0, r_2 = 1, \ldots, r_m = m-1$,
where $m$ is $\lfloor \tfrac{1}{\eps} \rfloor - 1$.
Suppose~$p_j = 1+\eps + \delta$ for every job~$j$.  Under the SRPT
algorithm, assuming that~$\eps+\delta$ is sufficiently small, there
will be at most~2 active jobs at all times (the most recently released
jobs); using~\eqref{eq:obj2}, the total flow time of its schedule
is~$O(\tfrac{1}{\eps})$.  The SETF algorithm will not complete any
jobs until after time~$m$, so in each time interval $[j-1,j]$ there
are~$j$ active jobs.  Using~\eqref{eq:obj2} again, the total flow time
of SETF's schedule is $\Omega(\tfrac{1}{\eps^2})$.
\end{example}

Example~\ref{ex:setf} shows that SETF is not optimal, and it
draws a line in the sand: The best we can hope for is that the SETF
algorithm with a $(1+\eps)$-speed machine achieves total flow
time~$O(\tfrac{1}{\eps})$ times that suffered by the SRPT algorithm with a
unit-speed machine.  The main result of this section states that this
is indeed the case.
\begin{theorem}[\citet{KP00}]\label{t:kp00}
  For every input and $\eps > 0$, the total flow time of the schedule
  produced by the SETF algorithm with a machine with speed $1+\eps$
  is at most
\[
1 + \frac{1}{\eps}
\]
times that by the SRPT algorithm with a unit-speed machine.
\end{theorem}
Using the second version~\eqref{eq:obj2} of the objective function,
Theorem~\ref{t:kp00} reduces to the following pointwise (over time)
bound.
\begin{lemma}\label{l:kp00}
Fix $\eps > 0$.
For every input, at every time step~$t$, 
\[
|X_t| \le \left( 1 + \frac{1}{\eps} \right) |X^*_t|,
\]
where~$X_t$ and~$X^*_t$ denote the jobs active at time~$t$ under SETF
with a $(1+\eps)$-speed machine and SRPT with a unit-speed machine,
respectively.
\end{lemma}
In Example~\ref{ex:setf}, at time~$t=m$, $|X^*_t| = 1$ (provided
$\eps,\delta$ are sufficiently small) while $|X_t|
= m \approx \tfrac{1}{\eps}$.  Thus, every inequality used in the
proof of Lemma~\ref{l:kp00} should hold almost with equality for the
instance in Example~\ref{ex:setf}.  The reader is encouraged to keep
this example in mind throughout the proof.

To describe the intuition behind Lemma~\ref{l:kp00}, fix a time~$t$.
Roughly: 
\begin{enumerate}

\item SRPT must have spent more time processing the jobs of~$X_t \sm
  X^*_t$ than SETF (because SRPT finished them by time~$t$ while SETF
  did not).

\item SETF performed $1+\eps$ times as much job processing as SRPT,
  an $\eps$ portion of which must have been devoted to the jobs of
  $X^*_t$. 

\item Because SETF prioritizes the jobs that have been processed the
  least, it also spent significant time processing the jobs of~$X_t
  \sm X^*_t$.

\item SRPT had enough time to complete all the jobs of $X_t \sm X^*_t$
  by time~$t$, so there can't be too many such jobs.

\end{enumerate}
The rest of this section supplies the appropriate details.

\subsection{Proof of Lemma~\ref{l:kp00}: Preliminaries}

Fix an input and a time~$t$, with $X_t$ and $X^*_t$ defined as in
Lemma~\ref{l:kp00}.  Rename the jobs of $X_t \sm X^*_t =
\{1,2,\ldots,k\}$ such that $r_1 \ge r_2 \ge \cdots \ge r_k$.

Consider the execution of the SETF algorithm with a $(1+\eps)$-speed
machine.  We say that job~$\ell$ {\em interferes} with job~$j$ if
there is a time $s \le t$ at which~$j$ is active and~$\ell$ is
processed in parallel with or instead of~$j$.  The {\em interference
  set} $I_j$ of a job~$j$ is the transitive closure of the
interference relation:
\begin{enumerate}

\item Initialize $I_j$ to $\{j\}$.

\item While there is a job $\ell$ that interferes with a job
  of~$I_j$, add one such job to~$I_j$.

\end{enumerate}
In Example~\ref{ex:setf} with $t=+\infty$, the interference set of
every job is the set of all jobs (because all of the jobs are
processed in parallel at the very end of the algorithm).  If
instead~$t=m$, then $I_j=\{j,j+1,\ldots,m\}$ for each job $j \in \{1,2,\ldots,m\}$.

The interference set of a job is uniquely defined, independent of
which interfering job is chosen in each iteration of the while loop.
Note that the interference set can contain jobs that were completed by
SETF strictly before time~$t$.

We require several properties of the interference sets of the jobs in
$X_t \sm X^*_t$.  To state the first, define the {\em lifetime} of a
job~$j$ as the interval~$[r_j,\min\{C_j,t\}]$ up to time~$t$ during
which it is active.
\begin{proposition}\label{prop:inf1}
Let $j \in \{ 1,2,\ldots,k\}$ be a job of $X_t \sm X^*_t$.
The union of the lifetimes of the jobs in an interference set~$I_j$ is
the interval~$[s_j,t]$, where~$s_j$ is the earliest release time of a
job in~$I_j$.
\end{proposition}

\begin{proof}
One job can interfere with another only if their lifetimes overlap.
By induction, the union of the lifetimes of jobs in~$I_j$ is an
interval.  The right endpoint of the interval is at most~$t$ by
definition, and is at least~$t$ because job~$j$ is active at
time~$t$.  The left endpoint of the interval is the earliest time at
which a job of~$I_j$ is active, which is~$\min_{\ell \in I_j}
r_{\ell}$.
\end{proof}

Conversely, every job processed in the interval corresponding to an
interference set belongs to that set.

\begin{proposition}\label{prop:inf2}
Let $j \in \{ 1,2,\ldots,k\}$ be a job of $X_t \sm X^*_t$ and $[s_j,t]$
the union of the lifetimes of the jobs in $j$'s interference set~$I_j$.
Every job processed at some time $s \in [s_j,t]$ belongs to~$I_j$.
\end{proposition}

\begin{proof}
Suppose job~$\ell$ is processed at some time $s \in [s_j,t]$.
Since~$[s_j,t]$ is the union of the lifetimes of the jobs in~$I_j$,
$I_j$ contains a job~$i$ that is active at time~$s$.  
If $i \neq \ell$, then job~$\ell$ 
interferes with~$i$ and hence also belongs to~$I_j$.
\end{proof}

The next proposition helps implement the third step of the intuition
outlined in Section~\ref{ss:setf}.
\begin{proposition}\label{prop:inf3}
Let $j \in \{ 1,2,\ldots,k\}$ be a job of $X_t \sm X^*_t$.
Let~$w_{\ell}$ denote the elapsed time of a job~${\ell}$ under
  SETF by time~$t$.  Then $w_{\ell} \le w_j$ for every job $\ell$ in
  $j$'s interference set $I_j$.
\end{proposition}

\begin{proof}
We proceed by induction on the additions to the interference set.
Consider an iteration of the construction that adds a job~$j_1$
to $I_j$.  By construction, 
there is a sequence of already-added jobs $j_2,j_3,\ldots,j_p$ such that
$j_p = j$ and $j_i$ interferes with
$j_{i+1}$ for each $i=1,2,\ldots,p-1$.  (Assume that~$p > 1$; otherwise
we're in the base case where $j_1 = j$ and there's nothing to prove.)
As in Proposition~\ref{prop:inf1}, the union of the lifetimes of
the jobs $\{j_2,j_3,\ldots,j_p\}$ forms an interval $[s,t]$; the right
endpoint is~$t$ because $j_p = j$ is active at time~$t$.  By
induction, $w_{j_i} \le w_j$ for every $i=2,3,\ldots,p$.  
Thus, whenever~$j_1$ is processed in the interval $[s,t]$, there is an
active job with elapsed time at most~$w_j$.  By virtue of being
processed by SETF, the elapsed time of~$j_1$ at any such point in time
is also at most~$w_j$.  The job~$j_1$ must be processed at least once
during the interval $[s,t]$ (as the job interferes with $j_2$),
so its elapsed time by time~$t$ is at most~$w_j$.
\end{proof}

\subsection{Proof of Lemma~\ref{l:kp00}: The Main Argument}

We are now prepared to implement formally the intuition outlined in
Section~\ref{ss:setf}.  

Fix a job $j \in X_t \sm X^*_t$; recall that
$X_t \sm X^*_t = \{1,2,\ldots,k\}$, with jobs indexed in nonincreasing
order of release time.  Let~$I_j$ denote the corresponding
interference set and $[s_j,t]$ the corresponding interval in
Proposition~\ref{prop:inf1}.  As in Proposition~\ref{prop:inf3},
let~$w_i$ denote the elapsed time of a job~$i$ under SETF at time~$t$.
All processing of the jobs in~$I_j$ (by SETF or SRPT) up to time~$t$
occurs in this interval, and all processing by SETF in this interval
is of jobs in~$I_j$ (Proposition~\ref{prop:inf2}).  Thus, the
value~$w_i$ is precisely the amount of time devoted by SETF to the
job~$i$ in the interval~$[s_j,t]$.

During the interval~$[s_j,t]$, the SRPT algorithm (with a unit-speed
machine) spends at most $t-s_j$ time processing jobs, and in
particular at most $t-s_j$ time processing jobs of~$I_j$.  Meanwhile,
the SETF algorithm works continually over the interval $[s_j,t]$; at
all times $s \in [s_j,t]$ there is at least one active job
(Proposition~\ref{prop:inf1}), and the SETF algorithm never idles with
an active job.  Thus SETF (with a $(1+\eps)$-speed machine)
processes~$(1+\eps)(t-s_j)$ units worth of jobs in this interval, and
all of this work is devoted to jobs of~$I_j$
(Proposition~\ref{prop:inf2}).

Now group the jobs of~$I_j$ into three categories:
\begin{enumerate}

\item Jobs $i \in I_j$ that belong to $X^*_t$ (i.e., SRPT has not
  completed~$i$ by time~$t$).

\item Jobs $i \in I_j$ that belong to $X_t$ but not $X^*_t$ (i.e.,
  SETF has not completed~$i$ by time~$t$, but SRPT has).

\item Jobs $i \in I_j$ that belong to neither~$X_t$ nor $X^*_t$ (i.e.,
both SETF and SRPT have completed~$i$ by time~$t$).

\end{enumerate}
The SRPT algorithm spends at least as much time as SETF in the
interval~$[s_j,t]$ processing category-2 jobs (as the former completes
them and the latter does not), as per the first step of the intuition
in Section~\ref{ss:setf}.
Both algorithms spend exactly the same
amount of time on category-3 jobs in this interval (namely, the sum of
the processing times of these jobs).  We can therefore conclude that
the excess time~$\eps(t-s_j)$ spent by the SETF algorithm (beyond that
spent by SRPT) is devoted entirely to category-1 jobs---the jobs
of~$X^*_t$ (cf., the second step of the outline in
Section~\ref{ss:setf}).  We summarize our progress so far in a
proposition.
\begin{proposition}\label{prop:inf4}
For every $j=1,2,\ldots,k$,
\[
\sum_{i \in I_j \cap X^*_t} w_i \ge \eps \cdot (t-s_j).
\]
\end{proposition}

The sum in Proposition~\ref{prop:inf4} is, at least, over the
jobs~$\{1,2,\ldots,j\}$.
\begin{proposition}\label{prop:inf5}
For every $j=1,2,\ldots,k$, the interference set~$I_j$ includes the
jobs $\{1,2,\ldots,j\}$.
\end{proposition}

\begin{proof}
Recall that the jobs $\{1,2,\ldots,k\}$ of $X_t \sm X^*_t$ are sorted
in nonincreasing order of release time.  Each job~$i=1,2,\ldots,j-1$
is released after job~$j$ and before job~$j$ completes (which is at
time~$t$ or later), and interferes with~$j$ at the time of its
release (as SETF begins processing it immediately).
\end{proof}

Combining Propositions~\ref{prop:inf4} and~\ref{prop:inf5},
we can associate unfinished work at time~$t$ for SETF with that of
SRPT:
\begin{corollary}\label{cor:inf}
For every $j=1,2,\ldots,k$,
\[
\sum_{i \in I_j \cap X^*_t} w_i \ge \eps \cdot \sum_{\ell=1}^j w_{\ell}.
\]
\end{corollary}
For example, taking~$j=1$, we can identify $\eps w_1$ units of time
that SETF spends processing the jobs of~$I_1 \cap X^*_t$ before
time~$t$.  Similarly, taking~$j=2$, we can identify~$\eps w_2$
different units of time that SETF spends processing the jobs of
$I_2 \cap X^*_t$: Corollary~\ref{cor:inf} ensures that the total
amount of time so spent is at least $\eps w_1 + \eps w_2$, with at
most $\eps w_1$ of it already accounted for in the first step.
Continuing with~$j=3,4,\ldots,k$, the end result of this process is a
collection $\{ \alpha(j,i) \}$ of nonnegative ``charges'' from
jobs~$j$ of $X_t \sm X^*_t$ to jobs~$i$ of $X^*_t$ that satisfies the
following properties:
\begin{enumerate}

\item For every $j=1,2,\ldots,k$, $\sum_{i \in X^*_t} \alpha(j,i) =
  \eps w_j$.

\item For every $i \in X^*_t$, $\sum_{j=1}^k \alpha(j,i) \le w_i$.

\item $\alpha(j,i) > 0$ only if $i \in I_j \cap X^*_t$.

\end{enumerate}
Combining the third property with Proposition~\ref{prop:inf3}:
\begin{equation}\label{eq:inf}
w_i \le w_j \text{ whenever } \alpha(j,i) > 0.
\end{equation}

We can extract from the $\alpha(j,i)$'s a type of network flow in a
bipartite graph with vertex sets $X_t \sm X^*_t$ and $X^*_t$.
Precisely, define the flow~$f^+_{ji}$ outgoing from~$j \in X_t \sm
X^*_t$ to $i \in X^*_t$ by
\[
f^+_{ji} = \frac{\alpha(j,i)}{w_j}
\]
and the flow~$f^-_{ji}$ incoming to~$i$ from~$j$ by
\[
f^-_{ji} = \frac{\alpha(j,i)}{w_i}.
\]
If we think of each vertex~$h$ as having a capacity of~$w_h$, then
$f^+_{ji}$ (respectively, $f^-_{ji}$) represents the fraction of~$j$'s
capacity (respectively, $i$'s capacity) consumed by the charge
$\alpha(j,i)$.  Property~\eqref{eq:inf} implies that the flow is
expansive, meaning that
\[
f^+_{ji} \le f^-_{ji}
\]
for every $j$ and $i$.

The first property of the $\alpha(j,i)$'s implies that there
are~$\eps$ units of flow outgoing from each~$j \in X_t \sm X^*_t$, for
a total of $\eps \cdot |X_t \sm X^*_t|$.  The second property implies
that there is at most one unit of flow incoming to each~$i \in X^*_t$,
for a total of at most~$|X^*_t|$.  Because the flow is expansive,
the total amount of flow incoming to~$X^*_t$ is at least 
that outgoing from~$X_t \sm X^*_t$, and so 
\[
|X^*_t| \ge \eps \cdot |X_t \sm X^*_t|.
\]
This completes the proof of Lemma~\ref{l:kp00}:
\[
|X_t| \le |X^*_t| + |X_t \sm X^*_t| \le |X^*_t| \cdot \left( 1 +
  \frac{1}{\eps} \right).
\]

\section{Loosely Competitive Algorithms}\label{s:lc}

An online algorithm with a good resource augmentation guarantee is
usually ``loosely competitive'' with the offline optimal algorithm,
which roughly means that, for every input, its performance is near-optimal
for most resource levels (cf., Figure~\ref{f:ra}(b)).  We illustrate
the idea using the online paging problem from Section~\ref{s:paging};
Exercise~\ref{exer:friedman} outlines an analogous result in the
selfish routing model of Section~\ref{s:sr}.

There is simple and accurate intuition behind the main result of this
section.  Consider a page request sequence~$z$ and a cache size~$k$.
Suppose the number of page faults incurred by the LRU algorithm is
roughly the same---within a factor of~2, say---with the cache
sizes~$k$ and~$2k$.  Theorem~\ref{t:lru_ra}, with $2k$ and $k$ playing
the roles of $k$ and $h$, respectively, then immediately implies that
the number of page faults incurred by the LRU algorithm with cache
size~$k$ is at most a constant (roughly~4) times that incurred by the
offline optimal algorithm with the same cache size.  In other words,
in this case the LRU algorithm is competitive in the traditional sense
(Definition~\ref{d:cr}).  Otherwise, the performance of the LRU
algorithm improves rapidly as the cache size is expanded from~$k$
to~$2k$.  But because there is a bound on the maximum fluctuation of
LRU's performance (between no page faults and faulting every time
step), its performance can only change rapidly for a bounded number of
different cache sizes.

Here is the precise statement, followed by discussion and a proof.
\begin{theorem}[\citet{Y02}]\label{t:young}
For every $\epsilon,\delta > 0$ and positive integer $n$, for every
page request sequence $z$, for all but a $\delta$ fraction of the
cache sizes~$k$ in $\{1,2,\ldots,n\}$, the LRU algorithm satisfies either:
\begin{enumerate}

\item $\perf(LRU,k,z) = O(\tfrac{1}{\delta} \log
\tfrac{1}{\epsilon}) \cdot \perf(FIF,k,z)$; or

\item $\perf(LRU,k,z) \le \epsilon \cdot |z|$.

\end{enumerate}
\end{theorem}

Thus, for every page request sequence~$z$, each cache
size~$k$ falls into one of three cases.
In the first case, the LRU algorithm with cache size~$k$ 
is competitive in the sense of Definition~\ref{d:cr}, with the number
of page faults incurred at most a constant (i.e., $O(\tfrac{1}{\delta}
\log \tfrac{1}{\epsilon})$) times the minimum possible.
In the second case, the LRU algorithm has a page fault rate of at
most~$\eps$, and thus has laudable performance in an absolute sense.
In the third case neither good event occurs, but
fortunately this happens for only a $\delta$ fraction of the possible
cache sizes.

The parameters $\delta$, $\epsilon$, and $n$ in
Theorem~\ref{t:young} are used in  the analysis only---no ``tuning'' of
the LRU algorithm is needed---and Theorem~\ref{t:young} holds
simultaneously for all choices of these parameters.  The larger the
fraction $\delta$ of bad cache sizes or the absolute performance bound
$\epsilon$ that can be tolerated, the better the relative performance
guarantee in the first case.

In effect, Theorem~\ref{t:young} shows that a resource augmentation
guarantee like Theorem~\ref{t:lru_ra}---an apples vs.\ oranges
comparison between an online algorithm with a big cache and an offline
algorithm with a small cache---has interesting implications for online
algorithms even compared with offline algorithms with the same cache
size.  This result dodges the lower bound on the competitive ratio of
the LRU algorithm (Section~\ref{ss:cr}) in two ways.  First,
Theorem~\ref{t:young} offers guarantees only for most choices of the
cache size~$k$; LRU might perform poorly for a few unlucky cache
sizes.  This is a reasonable relaxation, given that we don't expect
actual page request sequences to be adversarially tailored to the
choice of cache size.  Second, Theorem~\ref{t:young} does not insist
on good performance relative to the offline optimal algorithm---good
absolute performance (i.e., a very small page fault rate) is also
acceptable, as one would expect in a typical
application.\footnote{This may seem like an obvious point, but such
  appeals to good absolute performance are uncommon in the analysis of
  online algorithms.}

We proceed to the proof of Theorem~\ref{t:young}, which follows
closely the intuition laid out at the beginning of the section.

\begin{proof}
Fix a request sequence $z$ and values for the parameters $\delta$,
$\eps$, and $n$.
Let $b$ be a positive integer, to be chosen in due time.
The resource augmentation guarantee in
Theorem~\ref{t:lru_ra} states that, ignoring additive terms,
\begin{equation}\label{eq:young1}
\perf(LRU,k+b,z) \le \frac{k+b}{b+1} \cdot \perf(FIF,k,z),
\end{equation}
where~$k+b$ and~$k$ are playing the roles of~$k$ and~$h$ in
Theorem~\ref{t:lru_ra}, respectively.

There are two cases, depending on whether
\begin{equation}\label{eq:young2}
\perf(LRU,k+b,z) \ge \frac{1}{2} \cdot \perf(LRU,k,z)
\end{equation}
or
\[
\perf(LRU,k+b,z) < \frac{1}{2} \cdot \perf(LRU,k,z).
\]
Call a cache size~$k$ {\em good} or {\em bad} according to whether it
belongs to the first or second case, respectively.
For good cache sizes~$k$, chaining
together the inequalities~\eqref{eq:young1} and~\eqref{eq:young2} shows that
\begin{equation}\label{eq:young3}
\perf(LRU,k,z) \le 2 \cdot \frac{k+b}{b+1} \cdot
\perf(FIF,k,z),
\end{equation}
and hence LRU is competitive (with ratio $\tfrac{2(k+b)}{b+1}$) in the
sense of Definition~\ref{d:cr}.

Consider the set of bad cache sizes; for every such size, adding $b$
extra pages to the cache decreases the number of page faults
incurred by the LRU algorithm on~$z$ by at least a factor of~2.  If
there are at least $\ell$ bad cache sizes between 1 and $t-b$ for
some~$t$, then we can find $\ell/b$ bad cache sizes 
$k_1 < k_2 < \cdots < k_{\ell/b}$ in this
interval that are each at least $b$ apart (by taking every $b$th bad
cache size).\footnote{For clarity, we omit the appropriate ceilings
  and floors from fractions such as $\ell/b$.}
In this case, using that
$\perf(LRU,k,z)$ is nonincreasing in~$k$
(Exercise~\ref{exer:LRU_mono}), we have
\[
\perf(LRU,k_{i+1},z) < \frac{1}{2} \cdot \perf(LRU,k_{i},z)
\]
for each $i=1,2,\ldots,\ell/b$, where~$k_{(\ell/b)+1}$ should be
interpreted as $k_{\ell/b}+b \le t$.
Chaining all of these inequalities together yields
\[
\perf(LRU,t,z) < 2^{-\ell/b} \cdot \perf(LRU,1,z).
\]
Thus, once 
\begin{equation}\label{eq:young4}
\ell \ge b \cdot \log_2 \tfrac{1}{\epsilon},
\end{equation}
we have a page fault rate of at most~$\epsilon$:
\begin{equation}\label{eq:young5}
\perf(LRU,t,z) \le \epsilon \cdot |z|,
\end{equation}
where $|z|$ is the length of the request sequence $z$.

The time has come to instantiate the parameter~$b$.
Guided by our desire to have $\delta n$ bad cache sizes between~1 and
some number~$t$ force the condition that $\perf(LRU,k,z) \le \epsilon
|z|$ for all cache sizes $k \ge t$, we take $\ell = \delta n$.
The inequality~\eqref{eq:young4} then suggests taking $b = 
\delta n/\log_2 \tfrac{1}{\eps}$.

Cache sizes now fall into three categories:
\begin{enumerate}

\item Good cache sizes.  
By the inequality~\eqref{eq:young3} and our choice of~$b$,
\[
\perf(LRU,k,z) = O(\tfrac{1}{\delta} \log
\tfrac{1}{\epsilon}) \cdot \perf(FIF,k,z)
\] 
for every such cache size~$k$.

\item The smallest $\delta n$ bad cache sizes in $\{1,2,\ldots,n\}$.
  There is no performance guarantee for these cache sizes.

\item Bad cache sizes that are bigger than at least $\delta n$ other
  bad cache sizes.  Our choices of~$\ell$ and~$b$ ensure that the
inequality~\eqref{eq:young5} holds for such a cache size~$k$, with
\[
\perf(LRU,k,z) \le \eps |z|.
\] 

\end{enumerate}
Cache sizes in the first and third categories meet the first and
second guarantees, respectively, of Theorem~\ref{t:young}.  Cache
sizes in the second category constitute at most a $\delta$ fraction of
the possible cache sizes, so the proof is complete.
\end{proof}

\section{Notes}

Resource augmentation was first stressed as a first-order analysis
framework by \citet{KP00}, although there were compelling examples
much earlier (such as Theorem~\ref{t:lru_ra}, which was proved by
\citet{ST85}).  The phrase ``resource augmentation'' was proposed
shortly thereafter, by~\citet{PSTW02}.

The competitive analysis of online algorithms, including the model and
results in Section~\ref{s:paging},
was developed by
\citet{ST85}.  A good general reference for the topic is the book by
\citet{BE98}.  Theorem~\ref{t:fif} is due to \citet{B67}.
See~\citet[\S 2.4]{Y91} for empirical comparisons of the
FIF, LRU, and FIFO cache replacement policies on benchmark page
request sequences.

The selfish routing model described in Section~\ref{s:sr} was defined
by \citet{W52}.  Existence and uniqueness of equilibrium flows
(see footnote~6)
was proved by \citet{BMW56}; see also
\citet{rg}.  The price of anarchy was defined, in a different context,
by \citet{KP99}.  Theorem~\ref{thm:rt} and the extension in
Exercise~\ref{exer:mc} were proved by \citet{RT00}.  
The consequent loosely competitive bound
(Exercise~\ref{exer:friedman}) was proved by \citet{F04}.

\citet{PST03} is a good reference on the competitive analysis of
online scheduling algorithms; it includes a figure that inspired
Figure~\ref{f:ra}.  The optimality of SRPT (Exercise~\ref{exer:srpt})
was first proved by \citet{S68}.  Theorem~\ref{t:kp00} is by
\citet{KP00}, as is Exercise~\ref{exer:idle}.  One solution to
Exercise~\ref{exer:nonclairvoyant} appears in \citet{MPT94}.  There
are several more recent and sophisticated resource augmentation
guarantees for more complex scheduling problems, for example with
multiple machines, jobs with different priorities, and preemptions
replaced by a small number of rejections.
Good entry points to this literature include
\citet{IMP11}, \citet{AGK12}, and \citet{thang}.

The concept of a loosely competitive online algorithm 
is due to \citet{Y94} and 
Theorem~\ref{t:young} is from \citet{Y02}.

\section*{Acknowledgments}

I thank J\'er\'emy Barbay, Feder Fomin, Kirk Pruhs, Nguyen Kim Thang,
and Neal Young for helpful comments on a preliminary draft of this
chapter.

\section*{Exercises}

\begin{enumerate}

\item \label{exer:LRU_mono}
Prove that for every cache size $k \ge 1$ and every page sequence
$z$, 
\[ 
\perf(LRU,k+1,z) \le \perf(LRU,k,z).
\]

\item \label{exer:fifo}
Prove that Theorems~\ref{t:lru_ra} and~\ref{t:young} hold
also for the FIFO caching policy.

\item \label{exer:tight}
Prove a lower bound for all deterministic online algorithms that
matches the upper bound for LRU in
Theorem~\ref{t:lru_ra}.  That is, 
for every choice of $k$ and $h \le k$, every constant $\alpha <
\tfrac{k}{k-h+1}$, and every deterministic online
paging algorithm~$A$, there exist arbitrarily long sequences $z$ such
that $\perf(A,k,z) > \alpha \cdot \perf(FIF,h,z)$.

\item \label{exer:mc} Consider a {\em multicommodity} selfish routing
  network~$G=(V,E)$, with source vertices $s_1,s_2,\ldots,s_k$, sink
  vertices $t_1,t_2,\ldots,t_k$, and traffic rates
  $r_1,r_2,\ldots,r_k$.  A flow now routes, for each~$i=1,2,\ldots,k$,
  $r_i$ units of traffic from $s_i$ to $t_i$.  In an equilibrium
  flow~$f$, all traffic from $s_i$ to $t_i$ travels on $s_i$-$t_i$
  paths~$P$ with the minimum-possible length
  $\sum_{e \in P} c_e(f_e)$, where~$f_e$ denotes the total amount of
  traffic (across all source-sink pairs) using edge~$e$.

State and prove a generalization of Theorem~\ref{thm:rt} to
multicommodity selfish routing networks.

\item \label{exer:bicrit} Deduce Corollary~\ref{cor:bicrit} from
Theorem~\ref{thm:rt}.

\item \label{exer:friedman}
This problem derives a loosely
competitive-type bound from a resource augmentation bound in the
context of selfish routing (Section~\ref{s:sr}).
Let $\pi(G,r)$ denote the ratio of the costs of equilibrium
flows in~$G$ at the traffic rates $r$ and $r/2$.  By
Theorem~\ref{thm:rt}, the price of anarchy in the network $G$ at rate
$r$ is at most $\pi(G,r)$.
\begin{itemize}

\item [(a)] Use Theorem~\ref{thm:rt} to prove that, for every selfish
  routing network $G$ and traffic rate $r > 0$, and for at least an
  $\alpha$ fraction of the traffic rates $\hat{r}$ in $[r/2,r]$, the
  price of anarchy in $G$ at traffic rate $\hat{r}$ is at most
  $\beta \log \pi(G,r)$ (where $\alpha,\beta > 0$ are constants,
  independent of $G$ and $r$).

\item [(b)] Prove that for every constant $K > 0$, there exists a
  network $G$ with nonnegative, continuous, and nondecreasing edge
  cost functions and a traffic rate $r$ such that the price of anarchy
  in $G$ is at least $K$ for every traffic rate $\hat{r} \in [r/2,r]$.

\vspace{.25\baselineskip}
[Hint: use a network with many parallel links.]

\end{itemize}

\item \label{exer:srpt} Prove that the shortest remaining processing
  time (SRPT) algorithm is an optimal algorithm for the problem of
  scheduling jobs on   a single machine (with preemption allowed) to
  minimize the total flow time.

\item \label{exer:nonclairvoyant} Prove that for every constant
  $c > 0$, there is no non-clairvoyant deterministic online algorithm
  that always produces a schedule with total flow time at most~$c$
  times that of the optimal (i.e., SRPT) schedule.

\item \label{exer:idle}
Consider the objective of minimizing the maximum idle time of a
  job, where the {\em idle time} of job~$j$ in a schedule is $C_j-r_j-\tfrac{p_j}{s}$, where $C_j$ is
  the job's completion time, $r_j$ is its release time, $p_j$ is its
  processing time, and $s$ is the machine speed.  Show that the
  maximum idle time of a job under the SETF algorithm with a
  $(1+\eps)$-speed machine is at most $\tfrac{1}{\eps}$ times that in
  an optimal offline solution to the problem with a unit-speed machine.

\vspace{.25\baselineskip}

\noindent
[Hint: Start from Proposition~\ref{prop:inf3}.]

\end{enumerate}

\end{document}